%% file: paper.tex
\documentclass{paper}
\arxivtrue

\renewcommand{\sysname}{Matrix\xspace}
\begin{document}

\title{The Case for Accelerating BFT Protocols Using In-Network Ordering}

\author{Guangda Sun}
\author{Xin Zhe Khooi}
\author{Yunfan Li}
\author{Mingliang Jiang}
\author{Jialin Li}
\affil{National University of Singapore}

\maketitle

\input{abstract}
\input{intro}

\input{review}
\input{network}

\input{design}

\input{protocol}

\input{eval}

\input{relwk}
\input{conclude}

\bibliographystyle{plain}
\bibliography{paper}

\clearpage
\appendix

\input{protocol-more}
\input{proof}

\end{document}

%% file: abstract.tex
\begin{abstract}

Mission critical systems deployed in data centers today are facing more sophisticated failures. 
Byzantine fault tolerant (BFT) protocols are capable of masking these types of failures, but are rarely deployed due to their performance cost and complexity.
In this work, we propose a new approach to designing high performance BFT protocols in data centers. 
By re-examining the ordering responsibility between the network and the BFT protocol, we advocate a new abstraction offered by the data center network infrastructure. 
Concretely, we design a new authenticated ordered multicast primitive (\net) that provides transferable authentication and non-equivocation guarantees.
Feasibility of the design is demonstrated by two hardware implementations of \net -- one using HMAC and the other using public key cryptography for authentication -- on new-generation programmable switches.
We then co-design a new BFT protocol, \sysname, that leverages the guarantees of \net to eliminate cross-replica coordination and authentication in the common case. 
Evaluation results show that \sysname outperforms state-of-the-art protocols on both latency and throughput metrics by a wide margin, demonstrating the benefit of our new network ordering abstraction for BFT systems.

\end{abstract}

%% file: intro.tex
\section{Introduction}

Online services today are commonly deployed in large data centers and rely on fault-tolerance protocols to provide high availability in the presence of failures.
An important class of fault-tolerance protocols is state machine replication (SMR).
SMR protocols~\cite{paxos, pmmc, vr, zab, raft} have long been deployed in production systems~\cite{pml, chubby, zookeeper, spanner} to ensure a set of distributed nodes behaves like a single, always available state machine, despite failures of individual machines.
These protocols, however, can only tolerate node crash failures.
In reality, systems running in data centers are facing more sophisticated failures.
This is particularly relevant today, as permissioned blockchain systems~\cite{hyperledger, bidl, quorum} are increasingly being deployed in data centers for applications such as trading~\cite{sgx, asx}.
These systems require tolerance to adversarial nodes and attacks. 
Recent work~\cite{bidl} has shown that fault tolerance is becoming their main performance bottleneck.

Numerous Byzantine fault tolerant (BFT) protocols~\cite{pbft, zyzzyva, hotstuff, honeybadger, minbft, a2m, pompe} have been proposed to handle arbitrary node failures.
Their strong failure models, however, come with significant performance implications.
BFT protocols typically incur rounds of replica communication coupled with expensive cryptographic operations, resulting in low system throughput and high request latency.
To obtain higher throughput, many BFT protocols rely on heavy request batching which leads to long end-to-end decision latency -- often in the range of tens of milliseconds.
Unfortunately, such latency overheads are prohibitive for modern data center applications with stringent service-level objectives (SLOs).
Speculative BFT protocols such as Zyzzyva~\cite{zyzzyva} offer better commitment latency.
However, even a single faulty replica would eliminate the performance benefit of these latency-optimized protocols.

In this paper, we propose a new approach to building high-performance BFT protocols in data centers.
We observe that traditional BFT protocols are designed with minimum assumptions about the underlying network: The network only provides best effort message delivery.
The result of this weak network model is that application-level protocols are responsible for enforcing all correctness properties such as total ordering, durability, and authentication.
Our key insight is that strengthening the network model, in which the network can guarantee \textit{ordered message delivery}, can lead to reduced complexity and performance overhead of BFT protocols.
Prior work~\cite{nopaxos, eris} have shown promises of an in-network sequencing approach for crash fault-tolerant systems.
However, existing approaches fail to work once Byzantine failures are introduced: Faulty nodes can disseminate conflicting message orders, and the network sequencer may equivocate by assigning different sequence numbers to each replica.
In this work, we propose a new network-level primitive, \net, that addresses the above challenges.
\net ensures that correct receivers always deliver multicast messages in the same order, even in the presence of Byzantine participants.
The key property offered by \net is \textit{transferable in-network authentication} -- receivers can verify that a multicast message is properly delivered by \net, and they can prove authenticity of the message to other receivers in the system.
We additionally propose a mixed failure model~\cite{xft, hybrid-model} in which the network infrastructure can be either crash- or Byzantine-faulty.
For deployments that trust the network infrastructure, \net provides ordering guarantees with minimum network-level overhead.
For systems that need to tolerate Byzantine network devices, \net uses a simple round of cross-receiver communication to handle equivocating sequencers.

We demonstrate the feasibility of our approach by implementing our \net primitive in off-the-shelf programmable switches.
The switch data plane performs both \textit{sequencing} and \textit{authentication} for \net messages.
While packet sequencing is relatively straightforward, generating secure authentication codes -- a complex mathematical procedure -- is a major challenge given the switch's limited resources and computational constraints.
We propose two implementations of in-switch message authentication, each with trade-offs among switch resource utilization, performance, and scalability.
The first variant implements message authentication code (HMAC) vectors in the switch ASICs using SipHash~\cite{siphash} as the hashing algorithm.
The second variant generates signatures using public-key cryptography.
Due to hardware constraints, direct in-switch implementation of cryptographic algorithms such as RSA~\cite{rsa} and ECDSA~\cite{ecdsa} remains infeasible.
We design a new heterogeneous switch architecture that tightly couples FPGA-based cryptographic accelerators to the switch pipelines.
Our design enables efficient in-network processing and signing of \net messages, scales linearly with the number of cryptographic accelerators attached, and requires minimum hardware resources on the switch data plane.

Leveraging the strong properties provided by \net, we co-design a new BFT protocol, \sysname.
In the common case, \sysname replicas rely on the ordering guarantees of \net to commit client requests in a \textit{single round trip}, eliminating all cross replica communication and authentication.
Moreover, \sysname stays in this fast path protocol even in the presence of (up to $f$) faulty replicas, while requiring the theoretical minimum replication factor ($3f+1$).
In the rare case of network failures, we design efficient protocols to handle packet drops and faulty switch sequencers while guaranteeing correctness.
By evaluating against state-of-the-art BFT protocols, we show that \sysname can improve protocol throughput by up to 3.4$\times$ and end-to-end latency by 42$\times$, demonstrating the benefit of our authenticated in-network ordering approach for BFT systems.

%% file: review.tex
\section{Background}
\label{sec:background}

In this section, we give an overview of state-of-the-arts BFT protocols.
We then review recent proposals that use in-network ordering to accelerate crash fault-tolerant systems.


\begin{table*}[!ht]
\centering
\small
\begin{tabular}{@{}lcccccccc@{}}
\toprule
 & & PBFT~\cite{pbft} & Zyzzyva~\cite{zyzzyva} & SBFT~\cite{sbft} & HotStuff~\cite{hotstuff} & A2M~\cite{a2m} & MinBFT~\cite{minbft} & \sysname \\ 
\midrule
Replication Factor & & $3f+1$ & $3f+1$ & $3f+1$ & $3f+1$ & $2f+1$ & $2f+1$ & $3f+1$ \\
Bottleneck Complexity & & $O(N)$ & $O(N)$ & $O(N)$ & $O(N)$ & $O(N)$ & $O(N)$ & $O(1)$ \\ 
Authenticator Complexity & & $O(N^2)$ & $O(N)$ & $O(N)$ & $O(N)$ & $O(N^2)$ & $O(N^2)$ & $O(N)$ \\ 
Message Delay & & 5 & 3 & 6 & 4 & 5 & 4 & 2 \\
\bottomrule
\end{tabular}
\caption{
Comparison of \sysname to state-of-the-art BFT protocols.
Here, \textit{bottleneck complexity} denotes the number of messages the bottleneck replica needs to process; \textit{authenticator complexity} shows the total number of signatures processed by all replicas.
 }
\label{tab:bft-protocols}
\end{table*}

\subsection{State-of-the-Art BFT Protocols}
\label{sec:background:protocols}


There has been a long line of work on BFT state machine replication (SMR) protocols.
\autoref{tab:bft-protocols} presents a summary of the properties and comparison of some recent representative protocols.
PBFT~\cite{pbft} is the first practical BFT protocol that tolerates up to $f$ Byzantine nodes, requiring at least $3f + 1$ replicas which has been shown to be theoretical lower bound~\cite{async-consensus}.
In PBFT, client requests are committed in five message delays:
clients send requests to a primary replica; 
the primary replica sequences and forwards the requests to the backup replicas;
backup replicas authenticate the requests and broadcast their acceptance;
after a replica receives quorum acceptance, it broadcasts a commit decision;
replicas execute the request and reply to the client once they collect quorum commit decisions.
As replicas exchange messages in an all-to-all fashion, each replica processes $O(N)$ messages, making the authenticator complexity $O(N^2)$.

Zyzzyva~\cite{zyzzyva} speculatively executes client requests before they are committed to reduce the communication overhead.
The protocol includes a fast path with three message delays when clients receiving matching replies from \textit{all} replicas, and otherwise a slow path with at least five message delays.
The primary replica in Zyzzyva still sends signed messages to all back up replicas ($O(N)$), but with all-to-all communication removed, the authenticator complexity is reduced to $O(N)$.
Replicas in Zyzzyva may need to rollback speculatively executed operations during view changes.



Unlike Zyzzyva which pushes the responsibility of collecting authenticators to the client, SBFT~\cite{sbft} uses round-robin message collector among all replicas to remove all-to-all communication, similarly reducing authenticator complexity to $O(N)$.
SBFT also leverages threshold signatures to reduce message size, and simultaneously decreases the number of client replies to one per decision.


Many BFT protocols (\cite{pbft, zyzzyva, sbft}) use an expensive view change protocol to handle primary leader failure.
The standard view change protocol used in PBFT requires $O(N^3)$ message authenticators, limiting the overall scalability.
HotStuff~\cite{hotstuff} addresses this issue by adding an extra phase in normal operation. 
This design reduces the authenticator complexity of the leader failure protocol to $O(N)$, matching that of the normal case protocol.
In return, HotStuff adds a one-way network latency to the request commit delay.



\paragraph{BFT with trusted components.}
To reduce protocol complexity, a recent line of work~\cite{trinc,minbft,mft,ttcb,a2m} introduces \textit{trusted components} on each replica.
These trusted components can be implemented in a Trusted Platform Module (TPM)~\cite{tpm} or run in a trusted hypervisor, and are assumed to always behave correctly even if residing on Byzantine nodes.

A2M-PBFT-EA~\cite{a2m} uses an \textit{attested append-only memory} (A2M) to securely store operations as entries in a log.
Each A2M log entry is associated with a monotonically increasing, gap-less sequence number.
Once appended, a log entry becomes immutable, and its content can be attested by any node in the system.
A2M enables replicas to eliminate equivocation, and thus, reduce the replication factor to $2f+1$.
However, the protocol still incurs the same bottleneck complexity, authenticator complexity, and message latency as PBFT.
A later work, TrInc~\cite{trinc}, reduces the trusted component to a single counter, and leaves the log in untrusted memory.


MinBFT~\cite{minbft} features a message-based trusted primitive called Unique Sequential Identifier Generator (USIG). 
For each input message, USIG generates a unique identifier that is monotonic, sequential, and verifiable.
By authenticating USIG identifiers, MinBFT replicas can validate that all other replicas have received the same messages in the same order.
This property allows the protocol to reduce the message delay to four.
Unfortunately, MinBFT's authenticator complexity remains at $O(N^2)$.


\subsection{In-Network Ordering for CFT Protocols}

Another recent line of work~\cite{nopaxos, eris, specpaxos} proposes a new approach to designing crash fault tolerant (CFT) SMR protocols.
These systems move the responsibility of request ordering to the data center network, so application-level protocols only ensure durability of client operations.
This network co-design approach improves SMR protocol performance by reducing coordination overhead among servers to commit an operation.
For instance, NOPaxos~\cite{nopaxos} dedicates a programmable switch in the network as a sequencer.
The switch stamps sequence numbers to each NOPaxos request, and all replicas execute requests in the same sequence number order.
However, all these solutions only target CFT protocols and they fail to work if participants can exhibit Byzantine faults, e.g., a Byzantine node can impersonate the sequencer and broadcast conflicting message orders.

%% file: network.tex
\section{Authenticated In-Network Ordering}
\label{sec:network}

The core of a BFT SMR protocol is to establish a consistent order of requests even in the presence of failures.
Traditionally, this task is accomplished by explicit communication among the replicas, usually coordinated by a leader.
In this work, we argue for a new approach to improving the efficiency of BFT protocols, one which the responsibility of request ordering is moved to the network infrastructure.

\subsection{The Case for an Authenticated Ordering Service in Data Center Networks}
\label{sec:network:case}

To guarantee linearizability~\cite{linearizability}, BFT SMR protocols require all non-faulty replicas to execute client requests in the same order.
Due to the best-effort assumption about the network, an application-level protocol is fully responsible for establishing a total order of requests among the replicas.
Take PBFT~\cite{pbft} as an example:
The primary replica first assigns an order to client requests before broadcasting to backup replicas.
All replicas then use two rounds of communication to agree on this ordering while tolerating faulty participants.
As discussed in \cref{sec:background}, even adding trusted components to each replica does not alleviate the coordination and authentication overhead in BFT protocols:
replicas still require remote attestations to verify the received messages.

What if the underlying network can offer stronger guarantees?
Prior work~\cite{specpaxos, nopaxos, eris} have already demonstrated that in-network ordering, realized using network programmability~\cite{p4, tofino}, can offer compelling performance benefits to crash fault tolerant SMR protocols.
In this work, we argue that BFT protocols can similarly benefit from moving the ordering responsibility to the network.
At a high level, by offloading request ordering to the network, BFT replicas avoid explicit communication to establish an execution order, effectively reducing cross-replica coordination and authentication overhead.
This network ordering approach improves both protocol throughput -- less work is performed on each replica -- and latency -- fewer message delays are required to commit a request.


However, achieving message ordering in the network is more challenging in a BFT setting.
In this section, we discuss the main challenges faced by a BFT-based network ordering primitive and our approach to address them.

\paragraph{Why authenticated ordering in the network?}
In prior network ordering systems, the network primitive (such as the Ordered Unreliable Multicast in NOPaxos~\cite{nopaxos} and the multi-sequenced groupcast in Eris~\cite{eris}) is fully responsibility for assigning an order to client requests.
With non-Byzantine participants, this network-level ordering is the only request order observed by any replica.
Unfortunately, in a BFT deployment model, a faulty node can easily impersonate the network primitive and assign a conflicting message order, violating the ordering guarantee of the network layer.
To tolerate equivocating Byzantine participants, we augment the network primitive to provide an \textit{authentication} property:
Non-faulty replicas can independently validate that the received message order is indeed established by the network, not by any faulty node.
We elaborate on how such authentication can be implemented efficiently on commodity switch hardware in \cref{sec:design}.

\paragraph{Hybrid fault model and Byzantine network.}
If the network itself is Byzantine, it can equivocate by assigning different message orders to different replicas, violating the ordering guarantee.
In this work, we argue for a \textit{dual fault model} -- one fault model for the participating nodes and one for the network.
Our argument is inspired by prior work that proposes hybrid fault model~\cite{hybrid-model} and work~\cite{xft} that separates machine faults from network faults.
Specifically, we always assume a Byzantine failure mode for end-hosts.
The network infrastructure, on the other hand, can either be crash-faulty or Byzantine-faulty.
The benefit of our approach is that we offer flexibility with an explicit trade-off between fault tolerance and performance.
For deployments that trust the network to only exhibit crash and omission faults, our solution offers the optimal performance;
if the deployment assumes the network infrastructure can behave arbitrarily, we offer solution that tolerates Byzantine faults in the network, albeit taking a small performance penalty.

We further argue that a hybrid fault model, in which the network is assumed to be crash-faulty, is a practical option for many systems deployed in data centers.
Networking hardware presents a smaller attack surface and is less vulnerable to bugs compared to software-based components.
They are single application ASICs without sophisticated system software, and formal verification of their hardware designs are common practice.
Systems running in data centers also inherently place some trust in the data center hardware infrastructure.
It is also against the economic interest of data center operators to violate the trust of their provided services.

Our model resembles existing deployment options in the public cloud:
Only deployments that do not trust the cloud infrastructure run their virtual machines on instances with a Trusted Execution Environment (TEE) such as Intel SGX.
The majority of use cases place trust on cloud hardware and hypervisors; in return, they attain higher performance compared to their TEE counterpart.

\subsection{Authenticated Ordered Multicast}
\label{sec:network:tom}

We have so far argued for an authenticated ordering service in the network for BFT protocols.
To that end, we propose a new Authenticated Ordered Multicast (\net) primitive as a concrete instance of such model.
Similar to other multicast primitives (e.g., IP multicast), an \net deployment consists of one or multiple \net groups, each identified by a unique group address.
\net receivers can join and leave a \net group by contacting a membership service.
A sender sends an \net message to one \net group, and the network is responsible for routing the message to all group receivers.
Note that senders do not know the identity nor the address of individual receivers; they only specify the group address as the destination.

Unlike traditional best-effort IP multicast, \net provides a set of stronger guarantees, which we formally define here:

\begin{itemize}[leftmargin=*]
    \item \textbf{Asynchrony.} There is no bound on the delivery latency of \net messages.
    \item \textbf{Unreliability.} There is no guarantee that an \net message will be received by any receiver in the destination group.
    \item \textbf{Authentication.} A receiver can verify the authenticity of an \net message, i.e., the message is correctly processed by the \net network primitive.
    A correct receiver only delivers authentic \net messages.
    \item \textbf{Transferable Authentication.} If a receiver $r_1$ forwards an \net message to another receiver $r_2$, $r_2$ can independently verify the authenticity of the message.
    \item \textbf{Ordering.} For any two authentic \net messages $m_1$ and $m_2$ that destined to the same \net group $G$, all correct receivers in $G$ that receive both $m_1$ and $m_2$ deliver them in the same order.
    \item \textbf{Drop Detection.} Receivers can detect \net message drops in the network and deliver \drop. Formally, for any authentic \net message $m$, either 1) all correct receivers in the destination group $G$ delivers $m$ or a \drop for $m$ before delivering the next \tom message, or 2) none of the correct receivers in $G$ delivers $m$ or a \drop for $m$.
\end{itemize}

A key differentiating property of \net is the capability of receivers to independently verify the \textit{authenticity} of an \net message.
Authenticity in our case does not refer to the identity of the sender, which still requires end-to-end cryptography.
Instead, receivers can verify that a message is correctly processed by the \net primitive, and the ordering is not forged by other participant in the system.
The authentication capability is also transferable: an \net message can be relayed to any other receiver in the group, and they can independently verify the authenticity of the message.

Another important guarantee of \net is \textit{non-equivocation}, i.e., the primitive never delivers conflicting order of messages to non-faulty receivers in a group.
As we show in \cref{sec:design}, if the network is assumed to be Byzantine-faulty, the primitive requires additional confirmation among the group receivers to guarantee non-equivocation.
\net, however, does not guarantee reliable transmission.
The implication is that an \net message may be delivered by only a subset of the receivers in a group.
\net ensures that receivers who miss the message deliver a \drop.
The application-level BFT protocol is then responsible for reaching consensus on the fate of those dropped messages.

%% file: design.tex
\section{Design and Implementation of \net}
\label{sec:design}

In this section, we detail our design of the proposed \net network primitive on reconfigurable switches.
Even though ordering and drop detection can be realized using an in-switch sequencing design (\cref{sec:design:order}), achieving transferable authentication presents major technical challenges.
Here, we present two switch designs that overcome these challenges. 
The first design implements hash-based message authentication code (HMAC) vectors directly in the switch ASICs (\cref{sec:design:hmac}).
The second design proposes a new switch architecture that integrates cryptographic accelerators with the switch pipelines.
This architecture enables us to sign \net messages using public key cryptography in the network (\cref{sec:design:crypto}).

\subsection{Design Overview}
\label{sec:design:overview}

Our \net primitive design consists of three major components:
a network-wide configuration service, a programmable network data plane, and an application-level library running on \net senders and receivers.
Similar to IP multicast, receivers create and and join an \net group by contacting the configuration service through secure TLS channels.
The configuration service selects one programmable switch in the network as a \textit{sequencer} for the group, and instructs the sequencer switch to broadcast routing advertisement for the group address (e.g., using BGP).
Once the advertisement is propagated, \net messages destined to the group address will be forwarded to the sequencer switch.
The switch implements sequencing (\cref{sec:design:order}) and authentication (\cref{sec:design:hmac}, \cref{sec:design:crypto}) of \net packets in the data plane using the P4 language~\cite{p4}.
When sending an \net packet, the sender-side library generates a custom packet header (after the UDP header) that contains the \textit{group ID}, a \textit{sequence number}, an \textit{epoch number}, and an \textit{authenticator}.
Except the group ID, all remaining fields in the custom header are filled by the sequencer switch.
The receiver-side library verifies the authenticator and delivers \net messages in sequence number order.
For any gap in the number sequence, the receiver delivers a \drop.
For deployments that assume a Byzantine-faulty network, the receiver-side library additionally exchanges \confirm messages with other receivers in the group to tolerate sequencer equivocation (\cref{sec:design:order}).

\subsection{Message Ordering, Drop Detection, and Failure Handling}
\label{sec:design:order}

To establish a consistent ordering of \net messages, we leverage programmable switches to stamp monotonically increasing, gap-less sequence numbers to each \net packet.
The sequencer switch maintains a counter in a data plane register for each \net group.
When processing a \net packet, the switch uses the destination group address to locate the group counter register, increments the counter, and writes the counter value into the message header.
After the switch generates an authenticator for the packet, it uses its replication engine to multicast the stamped \net message to all \net receivers in the group.
As discussed in \cref{sec:design:overview}, receivers deliver authenticated \net messages in sequence number order.

Our sequencing approach enables \textit{drop detection}: When a receiver notices a gap in the message number sequence -- indicating some \net messages have been dropped in the network -- it delivers \drop{}s to the application.
This naive approach still works when messages are reordered in the network: A receiver simply ignores out-of-order messages.
However, it can lead to more frequent \drop{}s when packet reordering rate is high.
As an optimization, receivers buffer out-of-order \net messages and only deliver \drop after a timeout.

\paragraph{Tolerating Byzantine-faulty network.}
If the network infrastructure is non-Byzantine (\cref{sec:network:case}), a receiver can directly deliver authenticated \net messages in the sequence number order.
This delivery rule satisfies our \textit{ordering} property: Any two receivers are guaranteed to receive identical messages for each sequence number.
If a Byzantine-faulty network is assumed, the sequencer may equivocate by sending different request ordering to each receiver.
To tolerate sequencer equivocation, for each received \net message, a receiver broadcasts a signed confirmation that includes the sequence number and the message hash.
It ignores subsequent \net messages with the same sequence number.
A receiver only delivers an \net message after it collects enough matching confirmations (at least $2f+1$ where $f$ is the number of faulty receivers) for the message.
This strengthened delivery rules ensures our \textit{ordering} property in a Byzantine-faulty network, since no two non-faulty replicas can deliver distinct \net messages for the same sequence number.
When delivering, the receiver-side library delivers both the \net message and the collection of confirmation messages.
This entire message set serves as an \textit{ordering certificate} -- it can be used to prove the correct ordering of the associated \net message.
On the other hand, a single \net message alone can serve as the ordering certificate in a crash-faulty network.

\paragraph{Sequencer switch failover.}
The above scheme only ensures message ordering and drop detection when there is a single switch sequencing \net messages.
To tolerate faulty switches, \net receivers request the configuration service to fail over to a different sequencer for the group.
However, receivers in a group may not deliver the same set of messages from the failed sequencer when the new sequencer switch starts.
Furthermore, the previous sequencer may only suffer a transient fault; it may continue to process \net messages, concurrently with the new sequencer.
To properly handle a switch failover, the application-level protocol is responsible for reaching consensus on the set of messages delivered by the failed sequencer (\cref{sec:protocol:view-change}).
Once an agreement is reached, the receivers ask the configuration service to choose a new sequencer switch, and exchange the necessary authentication keys.
They then start delivering \net messages and \drop{}s for the new sequencer switch, ignoring messages from the old switch.

\subsection{HMAC-Based In-Network Authentication}
\label{sec:design:hmac}

Compared to message sequencing, generating \textit{secure} and \textit{transferable} authentication tokens for messages directly in network hardware raises bigger technical challenges.
Common authentication mechanisms such as Hash-based Message Authentication (HMAC) and public key cryptography involve long sequence of complex mathematical operations, which are challenging tasks for the resource-constrained network devices.
In this work, we demonstrate the feasibility of implementing in-network authentication on current generation programmable switches.

Our first design uses HMAC vector as the authentication token.
Specifically, when a receiver joins a \net group, it uses a key exchange protocol~\cite{keyexchange} to share a secret key with the current sequencer switch, facilitated by the configuration service.
The switch control plane installs the secret key of each receiver in the data plane.
When sending \net messages, the sender generates a digest of the message using a collision-resistant hash functions~\cite{md5}, and adds the digest to the message header.
The switch concatenates the stamped sequence number~\autoref{sec:design:order} and the message digest, and generates a vector of HMAC hashes, one for each receiver.
For each hash, the switch uses the digest of the message, the sequencer number, the secret key of the receiver, and a cryptographic hash function as the input.
It then writes the entire HMAC vector into the message header.
A receiver authenticates a \net message by checking if a locally computed HMAC hash matches the corresponding entry in the received HMAC vector.

\paragraph{In-switch HMAC implementation.}
Implementing an unforgeable HMAC requires access to collision resistant cryptographic hash functions.
Common options such as MD5~\cite{md5} and SHA-2~\cite{sha2} either require modular arithmetic or rounds of sequential heavy computations, making them impractical to implement on programmable switches.
Recent advancements -- HalfSipHash~\cite{2021-SPIN-SipHash} and P4-AES~\cite{2020-SPIN-AES}\footnote{We extend P4-AES to implement AES-CMAC~\cite{rfc4493}. However, for brevity yet without the loss of generality, we only focus on HalfSipHash for the remainder of the discussion.} -- demonstrate the feasibility of switch implementation of cryptographic hash functions with throughput up to $\sim$150 million hashes per second~\cite{2021-SPIN-SipHash}.
We therefore leverage HalfSipHash as a building block for our in-switch HMAC design. 

A cryptographic hash function, unfortunately, only solves half of the equation.
Implementing a HMAC vector in the switch data plane still requires us to overcome several technical challenges.
Firstly, HalfSipHash consumes non-negligible switch hardware resources (e.g., using all 12 pipeline stages), making parallel generation of HMAC vector difficult.
With careful optimizations, we managed to fit four parallel instances of HalfSipHash within a single switching pipeline.
We explicitly trade-off the number of pipeline passes for HMAC computation to achieve higher degree of parallelism, while maintaining high throughput (\cref{sec:eval:micro}).

Second, there exists data dependencies between HMAC computation and other packet processing logic such as sequencing. 
A naive combination of the logic would result in a dependency chain exceeding the hardware limit~\cite{2021-NetSoft-dSketch}.
To address the issue, we perform pipeline-folding (\autoref{fig:hmac}) to extend our computation beyond the available pipeline stages.
Furthermore, our approach decouples HMAC vector computation from other packet processing logic, resulting in a simpler and more modular design.

Lastly, the number of HMAC instances grows \textit{linearly} with the multicast group size.
By dedicating a pipeline for HMAC vector computation, our design can scale up by configuring additional ports in the HMAC pipeline as loopback ports.
We arrange receivers into logical groups of four, and perform HMAC computations for the different groups in parallel.
Receiver scalability is then only bounded by the configurable loopback ports (up to 72 receivers in our current design).

Putting things together, we illustrate our design in~\autoref{fig:hmac}.
We dedicate switch \textsc{pipe 1} for HMAC vector computation, while the senders and receivers are connected to \textsc{pipe 0}.
When an \net packet enters ingress of \textsc{pipe 0}, the switch sequences the packet and performs the necessary pre-processing (e.g., header initialization) for HMAC computation.
The packet is then multicasted to the dedicated loopback ports to perform HMAC computations in parallel.
Each HMAC computation requires 12 pipeline passes in \textsc{pipe 1}.
The resulting packet is then multicasted and forwarded to all receivers after post-processing (e.g., copying the HMAC result) at \textsc{pipe 0} egress. 

\begin{figure*}[t!]
\begin{subfigure}[b]{0.3\textwidth}
\includegraphics[width=\linewidth, height=2.5cm]{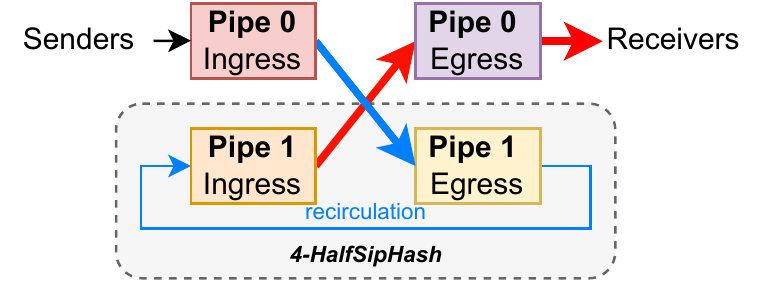}
\caption{Switch HMAC architecture}
\label{fig:hmac}
\end{subfigure}
\begin{subfigure}[b]{0.3\textwidth}
\centering
\includegraphics[width=\linewidth, height=2.5cm]{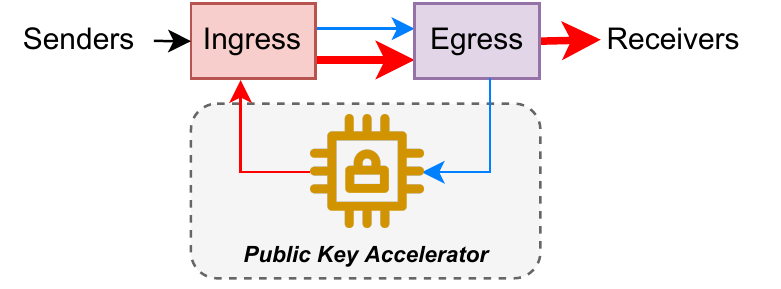}
\caption{Switch public key crypto architecture}
\label{fig:pkey}
\end{subfigure}
\begin{subfigure}[b]{0.3\textwidth}
\includegraphics[width=\linewidth]{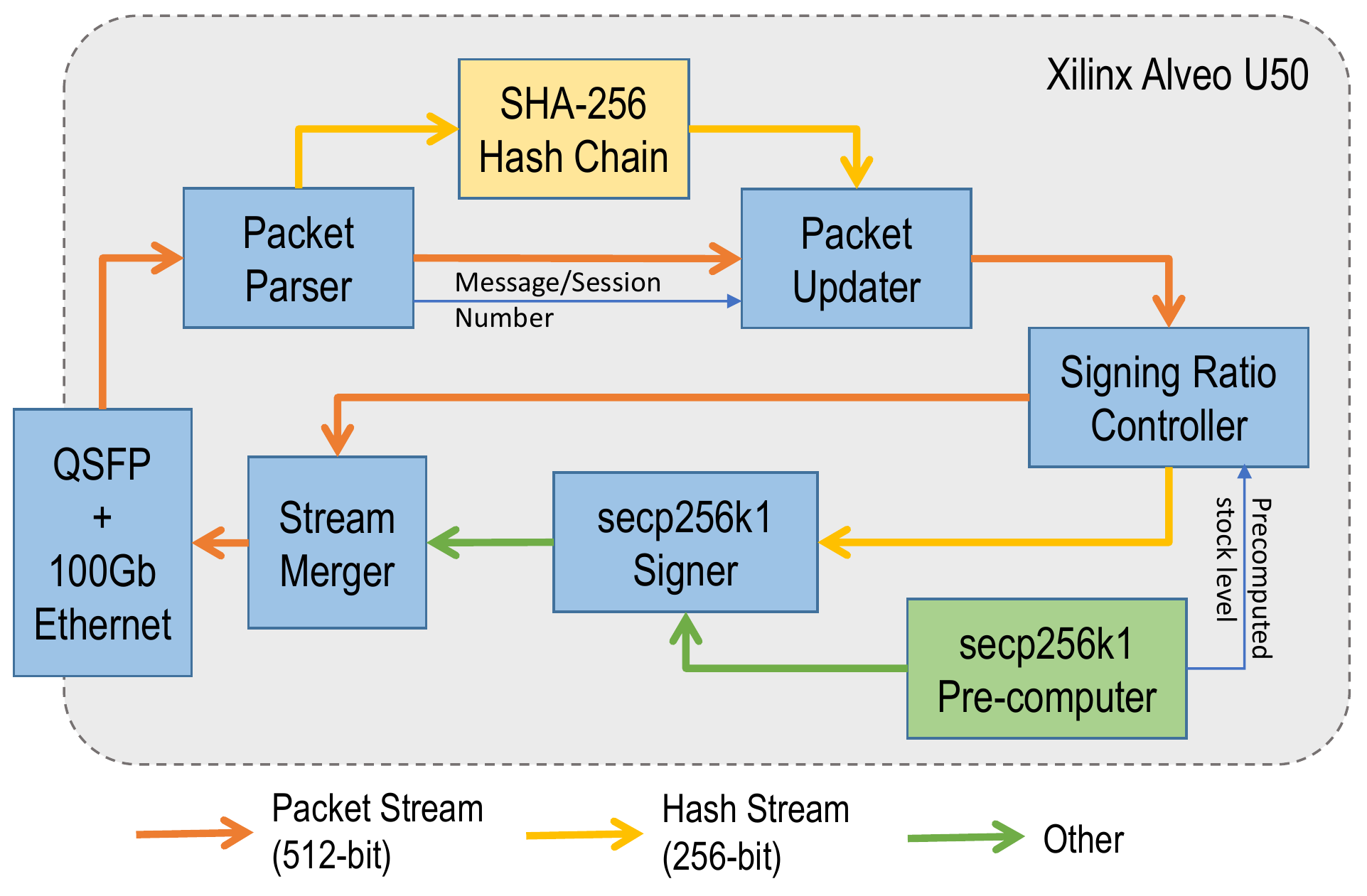}
\caption{FPGA crypto accelerator design}
\label{fig:fpga-arch}
\end{subfigure}
\caption{\autoref{fig:hmac} shows the folded pipeline design for implementing HMAC on a switch.~\autoref{fig:pkey} illustrates our new switch architecture that tightly couples a crypto accelerator with the switch pipeline. \textcolor{blue}{Blue arrows} denote unsigned \net packets whereas the \textcolor{red}{red arrows} represents signed \net packets.  
The \textbf{thick arrows} refers to multicast.
\autoref{fig:fpga-arch} shows the FPGA hardware design of our crypto accelerator.
}
\end{figure*}

\subsection{In-Network Public-Key Cryptography}
\label{sec:design:crypto}

One issue with our HMAC vector design is scalability:
the number of HMAC instances and the vector size increases linearly with the multicast group size.
The amount of switch resources and the packet header size thus bound the deployment scale.
For deployments that demand bigger group sizes, we propose a second design that implements message authentication using public-key cryptography.
Concretely, each sequencer switch maintains a secret private key, and a public key which it shares with the configuration service.
The service forwards the switch's public key to a receiver when it joins the \net group.
When processing \net messages, the switch produces a digital signature on the concatenated message digest and sequence number using a public key algorithm~\cite{ecdsa} and its private key.
It then writes the signature to the message header, before multicasting to all group receivers.
Receivers use the switch public key to verify the authenticity of the message.
Our in-network public-key approach is agnostic to the number of receivers: the switch only maintains one private key, and it writes a \textit{single} signature to an \net message, regardless of the number of receivers.


\paragraph{In-network cryptography design.}
Implementing public-key cryptography in a network switch is a daunting task.
The RSA~\cite{rsa} public-key algorithm requires modular exponentiation of large prime numbers.
Even with aggressive optimizations, calculating a signature still involves unbounded loops of multiplications and modulo operations.
The ECDSA~\cite{ecdsa} algorithm involves similar complexity, and additionally requires random number generations and multiplicative inverses.
None of these operations are supported on current generation programmable switches.
Due to strict timing, power, and resource constraints, future programmable switches are unlikely to add support for these computations.

To circumvent the limitations of existing switches, we propose a new switch architecture that adds a specialized cryptographic co-processor alongside the main switching chip, as shown in~\autoref{fig:pkey}.
The co-processor includes a simple processing element, dedicated fast memory, and one or more cryptographic accelerators.
It is connected to the switching chip through high-speed PCIe, interconnect (such as QuickPath Interconnect), or internal network links.
This architecture is highly feasible: PCIe links between the switching chip and the control plane CPU is commonplace in off-the-shelf switches, and switch vendors have roll-outed commercial products that couple powerful FPGAs with the switch ASIC~\cite{apsswitch}.
When the switch determines that a packet requires a cryptography operation, it writes the operation input (e.g., operation type and crypto key identifier) into pipeline metadata.
After egress pipeline processing, the switch submits both the metadata and the packet to the co-processor who performs the crypto operation.
The co-processor then writes the result into the header, and sends both the packet and the metadata back to the switch for further processing.

\paragraph{Cryptographic accelerator implementation.}

We implement an accelerator prototype that performs \net signature signing on a Xilinx Alveo U50 FPGA card~\cite{alveou50}.
\autoref{fig:fpga-arch} shows the high-level architecture of our hardware design.
The card is connected to one of the switch ports through a 100Gbps QSFP28 cable.
After an incoming packet is parsed by a parser module, a SHA-256~\cite{sha2} hashing module calculates a hash of the concatenated message digest and the stamped sequence number.
A signing module then generates a signature of the hash using the \texttt{secp256k1} elliptic curve~\cite{ecdsa}.
To reduce signing latency, part of the computation is done by a pre-compute module before the hash is received by the signer.
Finally, a stream merger module sends the signed packet back to the QSFP28 port.
All of the hardware modules, except the Xilinx QSFP28 hard IP, are developed in-house using a mixture of RTL and HLS.

To handle message rate higher than the signature generation throughput, we apply a hash chaining technique. 
Each \net packet is additionally stamped with a hash of the preceding packet in the number sequence.
For packets that do not have a signature, receivers wait until the next signed packet, and verify the entire batch by checking the chain in the reverse order.
This chaining optimization is implemented in the hardware hashing module.
Additionally, we design a signing ratio controller to govern the frequency of signature generation based on traffic rate and signing throughput.

%% file: protocol.tex
\section{The \sysname Protocol}
\label{sec:protocol}

Leveraging the authenticated ordering guarantee provided by our \net network primitive, we co-designed a new BFT protocol, \sysname, that commits client operations in a \textit{single RTT}, even in the presence of Byzantine replicas.
In the rare case of packet drops in the network, \sysname takes advantages of \net{}'s drop detection to efficiently recover from missing operations.
To tolerate faulty replicas and network sequencers, \sysname runs a unified view change protocol to replace the faulty participants.

\subsection{System Model}

We assumes a Byzantine failure model for clients and replicas of the protocol. 
Each \sysname instance is assigned a unique \net group address. 
When a replica joins a \sysname instance, it registers as a group receiver of the corresponding \net group.
Besides regular asynchronous unicast messages, nodes in the system can send \net messages to a \sysname group.
\net messages follow the properties as defined in \cref{sec:network:tom}.
We make standard cryptography assumptions: Nodes do not have enough computational resources to subvert the cryptographic hash functions, message authentication codes, and public-key cryptography algorithms we used in the protocol.
We also assume a strong adversary model: Byzantine nodes can collude with each other, but they cannot delay correct nodes indefinitely.

\sysname is a state machine replication~\cite{smr} protocol.
We assume all operations executed by the protocol are deterministic.
With less than $\lfloor \frac{n - 1}{3} \rfloor$ (Byzantine) faulty replicas in the system (where $n$ is the total number of replicas), \sysname guarantees linearizability~\cite{linearizability} of client operations.
Due to the impossibility of asynchronous consensus~\cite{flp}, \sysname only ensures liveness when the system is in period of synchrony.

\subsection{Protocol Overview}
\label{sec:proto:overview}

\sysname relies on the guarantees provided by the \net network primitive to achieve single RTT commitment in the \textit{common case}.
Specifically, clients multicast requests to \sysname replicas using \net.
In the absence of network-level anomalies (e.g., message drops and switch failures), all replicas delivers \net messages in the exact same order.
Crucially, such guarantee implies that replicas require no explicit communication to agree on the order of messages.
\sysname thus avoids the expensive cross-replica coordination and server signature signing/verification required by other BFT protocols.
Moreover, adversaries can not temper with the order of messages nor their content, as correct replicas can independently verify the authenticity and integrity of each \net message.
Once an \net message is delivered, replicas can immediately execute the request and respond to the client, resulting in a single phase fast path protocol.
As discussed in \cref{sec:design:order}, when delivering messages, the network primitive provides an \textit{ordering certificate}: the single \net message for a crash-faulty network, or the additional $2f+1$ confirmations for a Byzantine-faulty network.
Similar to previous speculative protocols~\cite{zyzzyva, nopaxos, specpaxos}, \sysname relies on clients to confirm operation durability.
However, since all correct replicas already established a total order of operations, \sysname does not require extra protocols to handle faulty replicas (Zyzzyva~\cite{zyzzyva}) or state divergence due to out-of-order speculative executions (Speculative Paxos~\cite{specpaxos}).

In the rare case where \net messages are dropped in the network, the \net primitive delivers \drop{}s to non-faulty replicas.
To handle \drop{}s, replicas only need to agree on whether to process or to skip the message, not the order of messages.
\sysname uses a BFT binary consensus protocol, driven by a leader replica, to reach this agreement.
In this protocol, the leader uses a single ordering certificate (received by any replica) to commit the corresponding message.
To permanently skip the message, the leader replica collects evidences from a quorum of replicas to form a \textit{drop certificate}, which non-leader replicas would verify before committing the message as a \nop.
We use drop certificates to prevent Byzantine replicas from delaying the agreement indefinitely.

A faulty \net sequencer may stop multicasting messages or deliberately equivocating or dropping messages.
To ensure progress, \sysname replicas request the network to replace the faulty sequencer.
Installment of a new sequencer indicates the start of a new epoch.
Correctness of the protocol requires replicas to agree on the set of messages processed in the last epoch before entering the new epoch.
To that end, each \sysname instance goes through a sequence of \textit{views};
each view is identified by a view number represented as a $\tupstart\textit{epoch-num}\com \textit{leader-num}\tupend$ 2-tuple.
We establish a \textit{total order} -- defined by the lexicographical order of the 2-tuple -- on the set of view numbers.
When the current leader replica has failed (or suspected to be failed) or an old epoch has ended, replicas advance the respective field in the view number, and use a view change protocol~\cite{vr, pbft, nopaxos, specpaxos} to reach agreement on the set of messages in the last view.

\sysname also includes a synchronization protocol that periodically synchronizes replica states and commits speculatively executed requests.
Details of the synchronization protocol, as well as formal safety and liveness proofs of \sysname, can be found in 
\ifarxiv
\autoref{sec:protocol-more}.
\else
the supplementary material~\cite{suppl}.
\fi

\subsection{Normal Operation}
\label{sec:protocol:normal}

We first consider the common case protocol in which \net messages, instead of \drop{}s, are delivered to \sysname replicas in a stable epoch.
A client \textit{c} requests execution of an operation \textit{op} by sending a signed message $\tupstart\req\com \textit{op}\com \textit{request-id}\tupend\sig{c}$ using the \net primitive, where \textit{request-id} is a client-generated identification to match replica replies. 
The message is processed by the network primitive (\cref{sec:network:tom}), and an ordering certificate (\oc) for the message is delivered to all replicas.
If the client does not receive replies in a timely manner, it uses regular unicast to send the request to all replicas directly (while keeps resending the request using \net).
Our view change protocol (\cref{sec:protocol:view-change}) ensures that the request is eventually committed even if the \net sequencer is faulty.

Replica $i$ verifies the \oc by authenticating the \net authenticator and checking the $2f+1$ matching confirmations (only for Byzantine-faulty network). It then adds the \oc to its log, speculatively executes \textit{op}, signs and replies $\tupstart\reply\com \textit{view-id}\com i\com \textit{log-slot-num}\com \textit{log-hash}\com \textit{request-id}\com \textit{result}\tupend\sig{i}$ to the client, where \textit{view-id} is the current view number, \textit{log-slot-num} is the log index the request occupies, \textit{log-hash} is a hash of the log up to the index, and \textit{result} is the execution result.
Similar to prior work~\cite{specpaxos}, we use hash chaining for $O(1)$ hash calculation.

Client $c$ waits for $2f+1$ replies from different replicas with valid signatures and matching \textit{view-id}, \textit{log-slot-num}, \textit{log-hash}, and \textit{result}. 
It then accepts the result in the reply.

\subsection{Handling Dropped Messages}
\label{sec:protocol:gap}

When a non-leader replica $i$ receives a \drop, it attempts to recover the missing message from the leader.
To do so, it sends a $\tupstart\query\com \textit{view-id}\com \textit{log-slot-num}\tupend$ to the leader.
\query messages require no signatures since they do not alter the state of a correct replica.
If the leader has the corresponding \oc, it responds with a $\tupstart\queryreply\com \textit{view-id}\com \textit{log-slot-num}\com \oc\tupend$.
Replica $i$ verifies the \oc and ensures the enclosed \net message is the missing message by checking the internal sequence number.
It then resumes normal operation.
Because \oc can be independently verified by any replica, \queryreply{}s also require no signatures.
Note that replica $i$ \textit{blocks} on waiting for the leader's response or a committed \nop before processing subsequent client requests, resending \query messages if necessary.

If the leader $l$ itself receives a \drop, it broadcasts a $\tupstart\gapfindmsg\com \textit{view-id}\com \textit{log-slot-num}\tupend\sig{l}$ to all replicas.
When replica $i$ receives a \gapfindmsg, it replies to the leader with either a $\tupstart\gaprecvmsg\com \textit{view-id}\com \textit{log-slot-num}\com \oc\tupend$ if it has received the ordering certificate, or a $\tupstart\gapdropmsg\com \textit{view-id}\com i\com \textit{log-slot-num}\tupend\sig{i}$ if it has also received a \drop for the message.
If a replica replies \gapdropmsg to a \gapfindmsg, it blocks until it receives the gap agreement decision (ignoring \queryreply{}s for the message).

Once the leader receives one \gaprecvmsg or $2f+1$ \gapdropmsg (including from itself), whichever happens first, it uses a \textit{binary} Byzantine agreement protocol, similar to PBFT~\cite{pbft}, to commit the decision. 
Replicas store the quorum prepare (\gapprepare) and commit (\gapcommit) messages it collects during the agreement protocol for future view changes.
We refer to the quorum \gapcommit{}s as a \textit{gap certificate}.
Once the decision is committed, replicas store either the \oc (message is received) or a \nop (message is permanently dropped) in their logs, and resume normal operation.
Details of this agreement protocol can be found in 
\ifarxiv
\autoref{sec:protocol-more:gap}.
\else
the supplementary material~\cite{suppl}.
\fi

\subsection{View Changes}
\label{sec:protocol:view-change}

We use a view change protocol, inspired by PBFT, to handle both leader failures and faulty \net sequencers.
The protocol guarantees that all committed operations (including \nop{}s) will carry over to the new view.
View changes can be initiated in two scenarios: 
(i) when a non-leader replica in view $\tupstart{}e, l\tupend$ fails to make progress in a gap agreement or state synchronization protocol after a timeout, it initiates a view change with a new view $\tupstart{}e, l+1\tupend$, and
(ii) when a replica receives a request message directly from the client (\cref{sec:protocol:normal}) but the request is not delivered by \net after a timeout, it initiates a view change with a new view $\tupstart{}e+1, l\tupend$.

For view changes that involves switching epochs, the protocol requires log consistency before entering the new epoch. 
To that end, we introduce an \textit{epoch certificate} consisting of $2f+1$ valid \epochstart messages from distinct replicas.
An epoch certificate is a proof of the agreed starting log position of the epoch.
We then define \textit{validity} of a replica log as the following: a replica log is valid if and only if (i) the starting log position of all epochs are supported by a valid \textit{epoch-cert}, and (ii) within each epoch $e$, all log positions are filled with either a valid \oc or a \nop supported by a gap certificate.

Our view change protocol is similar to the one in PBFT.
The main differences are the additional epoch certificates and the definition of log validity.
Details of the protocol can be found in 
\ifarxiv
\autoref{sec:protocol-more:viewchange}.
\else
the supplementary material~\cite{suppl}.
\fi

\subsection{Correctness}

Here, we sketch a proof of correctness for the \sysname protocol.
Complete safety and liveness proofs can be found in 
\ifarxiv
\autoref{sec:proof}.
\else
the supplementary material~\cite{suppl}.
\fi

The safety property we are proving is linearizability~\cite{linearizability}.
A key definition we use in our proof is \textit{committed operations}:
an operation is committed in a log slot if it is executed by $2f+1$ replicas with matching \textit{view-id}s and \textit{log-hash}es.

First, we show that within an epoch, if a request $r$ is committed at log slot $l$, no other request $r'$ ($r' \ne r$) or \nop can be committed at $l$.
Due to the guarantees of \net, no correct replicas will execute $r'$ at log slot $l$ given that some correct replica has already executed $r$, so $r'$ can never be committed at $l$.
To show that \nop cannot be committed, we prove by contradiction.
Assume a \nop is committed at $l$, some replica would have received a \gapdecision with $2f+1$ valid \gapdropmsg from the leader, and $2f$ \gapprepare containing a \textit{drop} decision from different replicas.
By quorum intersection, no replica can receive $2f$ distinct \gapprepare{}s with a \textit{recv} at $l$, making \textit{drop} the only possible decision.
Our view change protocol also ensures that the decision will persist in all subsequent views.
Moreover, $2f+1$ replicas have sent a \gapdropmsg, and at least $f+1$ of those replicas are non-faulty.
Since they block until they receive \gapcommit{}s and the only possible outcome is \textit{drop}, they will not execute $r$.
By quorum intersection, $r$ cannot be committed, leading to a contradiction.

Next, we show that within an epoch, if a request is committed at log slot $l$, all log slots before $l$ will also be committed.
A committed request at $l$ implies that at least $f+1$ correct replicas have matching logs up to $l$.
For each log slot before $l$, since the same request has been processed by $f+1$ correct replicas, we can apply the same reasoning as above to show that no other request or a \nop can ever be committed at that slot.

Lastly, we show that our view change protocol guarantees that correct replicas agree on all committed requests and \nop{}s across views, and that they start each epoch in a consistent log state.
The first point is easy to show given that our view change protocol merges $2f+1$ logs and using the quorum intersection principle.
To prove the second point, we only need to show that for each epoch $e$, all correct replicas end $e$ at the same log slot before starting $e+1$.
To enter epoch $e+1$, a correct replica needs a valid \textit{epoch-cert} for epoch $e+1$: $2f+1$ distinct \epochstart{}s with matching \textit{log-slot-num}.
By quorum intersection, no other \textit{epoch-cert} can exist for $e+1$ with a different \textit{log-slot-num}.
And since correct replicas verify \textit{epoch-cert} for every epoch during view changes, by induction, their logs will be in consistent state.

%% file: eval.tex
\section{Evaluation}
\label{sec:eval}

We implement \sysname and the \net library in $\sim$1600 lines of Rust code.
The HMAC version of the \net sequencer is implemented in ${\sim}1900$ lines of P4~\cite{p4} code compiled using the Intel P4 Studio version 9.7.0. 
The FPGA-based cryptographic accelerator comprises of ${\sim}1500$ lines of HLS C++/Verilog code, synthesized using the Xilinx Vivado Design Suite 2020.2~\cite{vivado}.

We compare \sysname against PBFT~\cite{pbft}, HotStuff~\cite{hotstuff}, Zyzzyva~\cite{zyzzyva}.
For a fair comparison, all protocols are implemented in the same Rust-based framework.
We implemented the event-driven version of HotStuff, which closely matches their open-source implementation~\cite{hotstuff-opensource}.
We also added batching support to PBFT, HotStuff, and Zyzzyva, following the respective batching technique proposed in the original work.
\sysname does not use any batching on the protocol level.
Only when using the public key cryptography variant of \net, \sysname replicas buffer packets until receiving a signed message from the network.
Unless specified otherwise, we ran all protocols on four replicas, tolerating one Byzantine failure.
As a baseline, we also implemented an unreplicated protocol that runs on one server, tolerating no failures.

\paragraph{Testbed.}
Our testbed consists of nine servers and a Xilinx Alveo U50 FPGA, all connected to an Intel Tofino-based~\cite{tofino} programmable switch.
Replicas ran on machines with dual 2.90GHz Intel Xeon Gold 6226R processors (32 physical cores), 256 GB RAM, and Mellanox CX-5 EN 100 Gbps NICs.
Clients ran on a machine with one 2.10GHz Intel Xeon Gold 6230 processor (20 physical cores), 96 GB RAM, and Mellanox CX-5 EN 100 Gbps NICs.

\subsection{Micro-benchmarks}
\label{sec:eval:micro}

We first conduct micro-benchmarks to measure the performance and resource usage of our \net network primitive.
We evaluate both the HMAC-based switch pipeline and the FPGA-based cryptographic accelerator.

\paragraph{HMAC-based in-switch authentication.}

\begin{table}[tb]
\centering
\begin{tabular}{p{12mm}|>{\centering\arraybackslash}m{8.5mm}>{\centering\arraybackslash}m{8.5mm}>{\centering\arraybackslash}m{8mm}>{\centering\arraybackslash}m{8mm}c}
 Module & Stages & Action Data & Hash Bit & Hash Unit & VLIW\\\hline
 Pipe 0 & 7 & 0.8\% & 2.0\% & 0\% &  3.4\%\\
 Pipe 1 & 12 & 12.8\% & 21.2\% & 77.8\% & 12.0\% \\
\end{tabular}
\caption{Switch resource usage of \net HMAC vector prototype}
\label{tab:sw-util}
\end{table}

\autoref{tab:sw-util} shows the switch resource usage of our in-network HMAC vector design~\autoref{sec:design:hmac}.
To measure the latency and sustainable throughput of our design, we generate 64 byte \net packets at line rate using the Tofino built-in packet generator.
The average latency for generating one HMAC vector (size of four) is $8.5\mu$s. 
Our switch design attains a maximum HMAC throughput of 5.7Mpps, which is around 22 million hashes per second. 

\paragraph{FPGA-based ECDSA accelerator.}

\begin{table}[tb]
\centering
\begin{tabular}{p{15mm}|cccc}
 Module & LUT & Register & BRAM & DSP\\\hline
 Pipeline & 0.91\% & 0.70\% & 2.12\% & 0.57\% \\
 Signer & 21.0\% & 19.4\% & 10.71\% & 28.52\% \\
 Total & 34.69\% & 29.22\% & 28.76\% & 29.16\% \\\hline
 Available & 870K & 1740K & 1.34K & 5.94K
\end{tabular}
\caption{FPGA resource usage of \net public-key cryptographic accelerator}
\label{tab:fpga-util}
\end{table}

Hardware resource usage of our crypto accelerator (\autoref{sec:design:crypto}) is summarized in \autoref{tab:fpga-util}.
Using RTL-level simulation at 230MHz clock rate, latency of packet processing without signature signing is 835ns, giving a maximum throughput of 1.20M packets/sec.
The \texttt{secp256k1} signer module incurs around 745ns additional latency. 
The burst signing rate is at 57M signatures/sec, while the sustained signing rate is around 81.78K signatures/sec, limited by the pre-calculation module.

We also measured end-to-end performance by generating \net packets from an end-host to the FPGA through the Tofino switch. 
To accurately measure accelerator latency, the switch adds hardware timestamps before forwarding to the accelerator and after receiving the reply.
Signing latency is measured at $2.36\mu$s (including the Ethernet module).
The maximum overall throughput is 1.20M packet/sec, and the sustainable signing throughput is 71.7K signatures/sec.

\subsection{Latency vs. Throughput}

\begin{figure}
    \centering
    \includegraphics[width=\linewidth]{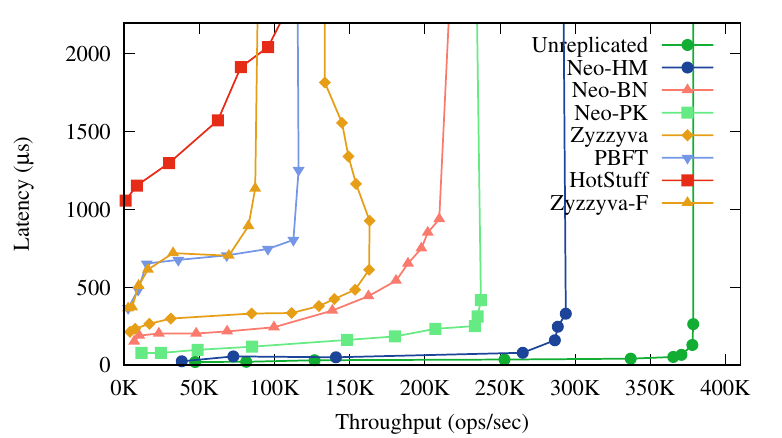}
    \caption{Comparing latency and throughput of \sysname and other BFT protocols. Neo-HM and Neo-PK are the HMAC and public key versions of \sysname, and Neo-BN uses the \net variant that tolerances a Byzantine network. 
    Zyzzyva-F is Zyzzyva with a non-responding Byzantine replica.}
    \label{fig:latency-throughput}
\end{figure}

We next evaluate the latency and throughput of \sysname and compare them to other BFT protocols.
Our focus here is protocol-level performance, so we ran a simple echo-RPC application: 
client requests are randomly generated strings, and the state machine running on the replicas simply echo the same string back to the client.
We used an increasing number of closed-loop clients, and measured the end-to-end latency and throughput observed by the clients.

As shown in~\autoref{fig:latency-throughput}, HMAC-based \sysname achieves higher maximum throughput than PBFT (2.5$\times$) and HotStuff (3.4$\times$).
More aggressive batching can further increase HotStuff's throughput to a level comparable to \sysname; however, its latency also increases to more than 10ms.
To commit a client operation, these protocols require explicit coordination among the replicas, with each message requiring expensive cryptographic operations.
\sysname, on the other hand, leverages guarantees of \net to eliminate coordination and cross-replica authentication overhead in the common case.
Comparing to Zyzzyva, \sysname still achieves 1.8$\times$ higher throughput.
Moreover, when one of the replicas become faulty, throughput of Zyzzyva drops by more than 54\%, while throughput of \sysname is unaffected. 
When using the public-key variant of \net, \sysname only suffers a 60K throughput decrease, despite requiring more expensive cryptographic operations.
It demonstrates the efficiency of our in-network crypto accelerator design.

\autoref{fig:latency-throughput} also shows the bigger benefit of \sysname --- latency.
HMAC-based \sysname outperforms PBFT in latency by 14.68$\times$, HotStuff by 42.28$\times$, and Zyzzyva by 8.56$\times$.
\sysname commits client operations in two message delays, while the other three protocols require at least three message delays with additional authentication penalties.
Using the public-key variant of \net adds about 55$\mu$s to the latency of \sysname. 
However, this version of \sysname still outperforms all the other protocols in latency by at least 2.7$\times$.


\paragraph{Tolerating Byzantine network.}
As discussed in \autoref{sec:design:order}, to tolerate Byzantine network sequencers, \net receivers exchange and authenticate confirmation messages.
This can lead to degraded throughput and latency compared to the non-Byzantine network variant.
\autoref{fig:latency-throughput} shows the performance of \sysname when tolerating a Byzantine network.
By batch processing of confirmation messages, \sysname minimizes the impact of the additional message exchanges, and is able to sustain a high throughput at the expense of higher latency. 
As shown in the figure, this \sysname variant still outperforms the other comparison protocols in both throughput and latency.

\subsection{Protocol Scalability}
\begin{figure}
    \centering
    \includegraphics[width=\linewidth]{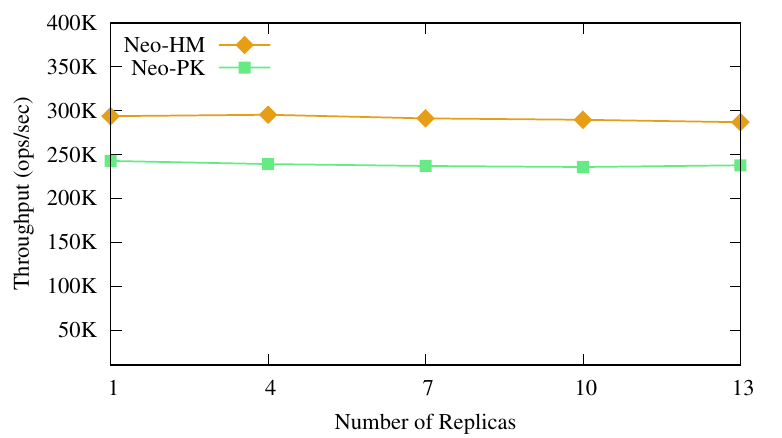}
    \caption{Measuring the throughput of public-key cryptography version of \sysname with increasing number of replicas}
    \label{fig:scalability}
\end{figure}

To evaluate the scalability of \sysname, we gradually increase the number of \sysname replicas (i.e., capable of tolerating more Byzantine failures), and measure the maximum sustainable throughput.
Due to the limited number of physical servers in our cluster, we deploy up to $13$ replicas, with $f$ replicas not running to simulate faulty nodes.
As shown in~\autoref{fig:scalability}, \sysname is able to scale to $13$ replicas with only a ${\sim}2.39$\% throughput drop.
\sysname replicas process a constant number of messages per client request, regardless of the replica count.
This allows \sysname to scale its performance almost linearly with more replicas.
Adding replicas, however, would increase the number of reply messages \sysname clients need to receive.
\sysname effectively shift the collector load to the client, which can naturally scale.

\subsection{Resilience to Network Anomalies}

\begin{figure}
    \centering
    \includegraphics[width=\linewidth]{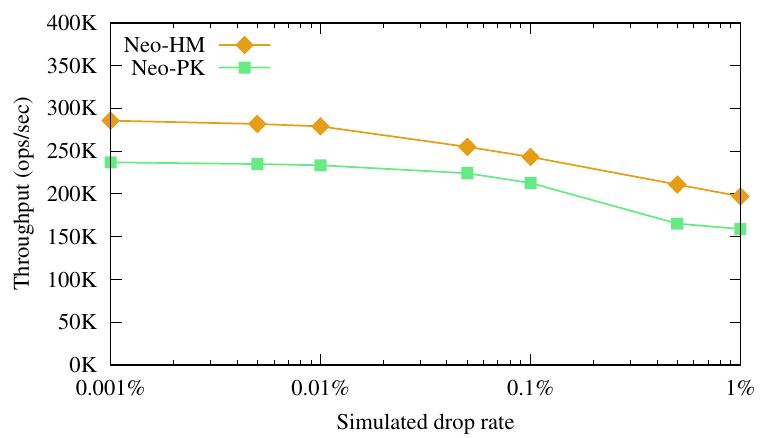}
    \caption{Throughput of \sysname with simulated packet drops}
    \label{fig:drops}
\end{figure}

When \net messages are dropped in the network, \sysname replicas coordinate to agree on the fate of the message.
To evaluate \sysname{}'s resilience to network anomalies, we simulate packet drops in the network, and measure the maximum throughput of \sysname.
As shown in \autoref{fig:drops}, throughput of \sysname is largely unaffected when moderate amount of packets are dropped.
This is due to the drop detection property of \net: non-faulty \sysname replicas can efficiently recover missing messages from each other, without the expensive agreement protocol.
When a higher percentage of packets are dropped (1\%), \sysname does suffer a more observable throughput drop.

\paragraph{Sequencer switch failover.}
When the sequencer switch becomes faulty, \sysname performs a view change and fails over to a different sequencer.
To understand the impact of a faulty sequencer, we measure the throughput of \sysname during a switch failover.
Throughput of \sysname dropped to zero after the sequencer became faulty.
After replicas detect the fault, they ran a view change protocol which finished in less than 200 $\mu$s.
They then inform the network controller to pick a new sequencer.
After rerouting is done (around 100 ms), throughput of \sysname quickly resumed to its peak.

\subsection{BFT Storage System Performance}


\begin{figure}
    \centering
    \includegraphics[width=\linewidth]{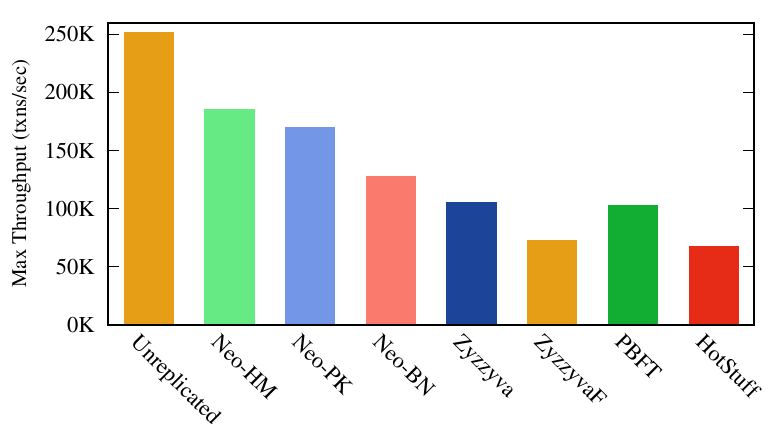}
    \caption{Comparing performance of replicated key-value store using YCSB}
    \label{fig:kv}
\end{figure}

Lastly, we evaluated the performance of \sysname when running more complex real-world applications, and compared against other protocols.
We developed an in-memory, B-Tree-based key-value store, and ran YCSB workload A with 100K records and 128-bytes fields.
Maximum YCSB throughput attained by each system is shown in~\autoref{fig:kv}.
\sysname achieved higher throughput than PBFT, HotStuff, and Zyzzyva when running a more complex application. 
This application requires protocols to handle larger requests than previous experiments, which reduces the efficiency of batching for Zyzzyva, PBFT and HotStuff. 
\sysname exploits its lower message complexity to attain higher performance.

%% file: relwk.tex
\section{Related Work}
\label{sec:relwk}

\paragraph{BFT protocols.}
As discussed in \cref{sec:background:protocols}, there has been a long line of work on designing practical BFT protocols~\cite{pbft, zyzzyva, sbft, hotstuff}.
These protocols guarantee correctness in an asynchronous network, and ensure liveness during weak asynchrony.
They all use a single leader node to coordinate ordering and agreement, and rely on view change (or similar) protocols to deal with faulty leaders.
Byzantine Paxos~\cite{bpaxos} and DBFT~\cite{dbft} propose a leaderless BFT design, but require a synchronous protocol that commits in $O(f)$ or more rounds.
HoneyBadger~\cite{honeybadger} attacks the weak synchrony assumption and provides optimal asymptotic efficiency.
However, it introduces $O(N^3)$ message complexity and five message delays.
\sysname leverages the guarantees of \tom to eliminate the leader and coordination overhead in the common case, leading to a bottleneck message complexity of $O(1)$ and two message delays to commit an operation.

\paragraph{BFT with trusted components.}
Recent work have proposed leveraging trusted components to improve BFT protocols~\cite{a2m, trinc, minbft, ttcb, hybster}.
Using a local trusted component on each replica enables these protocols to reduce the replication factor from $3f+1$ to $2f+1$.
However, since the trusted components are all local to replicas, they still necessitate coordination among replicas to commit each client operation.
Moreover, many of these protocols propose to implement trusted components in resource-constrained TPM hardware which significantly limit their performance.
\sysname implements its ordering service in the data center network.
Relying on authenticated network ordering, the protocol avoid all coordination in normal operation.
And by implementing on fast networking hardware, the service does not become the performance bottleneck.

\paragraph{Network ordering.}
A classic line of work in distributed computing proposes stronger network models, such as atomic broadcast~\cite{zab, abcast} and virtual synchrony~\cite{isis, virtsync}, to simplify distributed system designs.
These network primitives guarantees that a total order of messages are delivered to all broadcast receivers.
However, atomic broadcast and virtual synchrony do not offer performance benefits to distributed systems -- implementing them is equivalent to solving consensus~\cite{faildetect}.
NOPaxos~\cite{nopaxos} and Eris~\cite{eris} pioneered a weaker network model in which messages are delivered in a consistent order but reliable transmission is not guaranteed.
This weaker model can be efficiently implemented using programmable switches and other networking hardware.
NOPaxos proposes a Ordered Unreliable Multicast primitive for state machine replication, while Eris designs a multi-sequenced groupcast primitive for distributed transactions.
BIDL~\cite{bidl} uses sequencers to parallelize consensus and transaction execution in a permissioned blockchain system.
However, BIDL still uses traditional BFT protocols for consensus and its sequencer design does not improve performance of the BFT protocol itself.
The network sequencing approach has also been applied to other distributed system designs~\cite{corfu, vcorfu, tango}.
Our \net primitive was inspired from these work. 
Moreover, \net provides transferable authentication guarantee, which is crucial for BFT protocols.


%% file: conclude.tex
\section{Conclusion}
\label{sec:conclusion}

In this work, we presented a new approach to designing high performance BFT protocols in data centers.
Our work proposed a novel in-network authenticated ordering service. 
We demonstrated the feasibility of this design by implementing two network primitive variants, one using HMAC and another using public key cryptography for authentication.
We then co-designed a new BFT SMR protocol, \sysname, that eliminate cross-replica coordination and authentication in the common case.
\sysname outperforms state-of-the-art BFT protocols on both throughput and latency metrics.

%% file: protocol-more.tex
\section{Additional Protocol Details}
\label{sec:protocol-more}

\subsection{Gap Agreement Protocol}
\label{sec:protocol-more:gap}
When the leader in \sysname receives a \drop, it broadcasts a \gapfindmsg, and waits for either a single \gaprecvmsg or $2f+1$ \gapdropmsg{}s.
It then uses a binary Byzantine agreement protocol to commit the decision.
In this section, we describe the detail of the agreement protocol, which is omitted in the main paper.

The leader broadcasts a $\tupstart\gapdecision\com \textit{view-id}\com \textit{log-slot-num}\com \textit{decision}\tupend\sig{l}$, where \textit{decision} is either a single \gaprecvmsg or $2f+1$ \gapdropmsg{}s.
If a \gaprecvmsg is received, the leader first verifies the enclosed \oc following the same procedure as above.

When replica $i$ receives a \gapdecision, it verifies the enclosed \oc if the decision contains a \gaprecvmsg.
If the decision contains $2f+1$ \gapdropmsg, the replica verifies that all $2f+1$ messages are from distinct replicas, and their \textit{log-slot-num} matches the one in the \gapdecision.
It then broadcasts a $\tupstart\gapprepare\com \textit{view-id}\com i\com \textit{log-slot-num}\com \textit{recv-or-drop}\tupend\sig{i}$, where \textit{recv-or-drop} is a binary value indicating the decision, to all replicas.

Once replica $i$ receives $2f$ \gapprepare{}s from distinct replicas (possibly including itself) and it has received a validated \gapdecision with a matching decision from the leader, it broadcasts a $\tupstart\gapcommit\com \textit{view-id}\com \textit{log-slot-num}\com \textit{recv-or-drop}\tupend\sig{i}$ to all replicas.
The replica also stores the \gapdecision and $2f$ \gapprepare{} in its log for the view change protocol.

When replica $i$ receives $2f+1$ \gapcommit{}s from different replicas (possibly including itself), it stores either the \oc (if it hasn't done so) or a \nop to the log slot based on the decision, and resuming normal operation if it is blocking on a \queryreply or a gap agreement decision.
It also stores all $2f+1$ \gapcommit{}s in its log.
This quorum of \gapcommit{} will serve as a \textit{gap certificate} for the state synchronization and view change protocols.
In the rare case where the replica has already speculatively executed the request and the decision is a drop, it rolls back the application state to right before \textit{log-slot-num}, and re-executes subsequent requests in the log.

\subsection{View Change Protocol}
\label{sec:protocol-more:viewchange}

In this section, we provide the detail of the view change protocol, which is omitted in the main paper.

When replica $i$ initiates a view change, it broadcasts a $\tupstart\viewchange\com \textit{view-id}\com v'\com \textit{epoch-cert}\com \textit{log}\tupend\sig{i}$ to all other replicas, where \textit{view-id} is its current view number, $v'$ is the new view, and \textit{epoch-cert} contains an epoch certificate for each epoch it has started.

When leader replica $l$ of view $v'$ receives $2f$ \textit{valid} \viewchange messages for view $v'$ from different replicas, it merges the logs in the \viewchange messages as follows:
\begin{enumerate}
    \item It finds the \textit{view-id} with the largest epoch number $e$ that is supported by a \textit{epoch-cert}.
    \item If the replica has not started epoch $e$ yet, it finds a valid log that has started epoch $e$. 
    It then copies all requests and \nop{}s before the starting log position of epoch $e$ to its own log.
    \item From all valid logs that have started epoch $e$, it locates the log with the largest \textit{seq-num} in epoch $e$.
    It then copies all requests in epoch $e$ from the log to its own log.
    \item From any valid log that have started epoch $e$, it copies all \nop{}s in epoch $e$ from the log into its own log, possibly overwriting existing requests.
\end{enumerate}
Leader replica $l$ then broadcasts a $\tupstart\viewstart\com v'\com \textit{view-change-msgs}\tupend\sig{l}$ to other replicas, where \textit{view-change-msgs} contains the $2f$ \viewchange messages it uses to merge the log and the \viewchange message it would have sent for $v'$.

When replica $i$ receives a \viewstart message with $v'$ higher than its current view, it checks that \textit{view-change-msgs} are properly signed by $2f+1$ different replicas, they all contain the same next view number $v'$, and their logs are valid.
It then merges its log with logs in \textit{view-change-msgs} using the same procedure we described above.

If the view change does not involve a epoch switch, replica $i$ can immediately enter the new view.
Otherwise, it broadcasts a $\tupstart\epochstart\com e'\com \textit{log-slot-num}\tupend\sig{i}$ to all replicas, where $e'$ is the new epoch number, and \textit{log-slot-num} is the last log index after merging the logs during view change.
Once a replica receives $2f+1$ \epochstart messages from different replicas with $e'$ and \textit{log-slot-num} matching its own, it can enter the new view.
It also stores these \epochstart messages locally as an epoch certificate for future view changes.

\subsection{State Synchronization}
\label{sec:protocol-more:sync}

During normal operations, replicas execute client requests speculatively before they become durable.
A speculatively executed request might be overwritten due to the gap agreement or view change protocols, and the replica has to roll back application state and re-executes all subsequent requests in the log.
To further reduce the frequency of roll backs and the number of re-executions, we use a periodic synchronization protocol. 
The goal of the synchronization protocol is to produce a \textit{sync-point}, where all log entries before and including the \textit{sync-point} are \textit{committed}.
A committed log entry will never be overwritten or removed, and will be present in the log (at the same position) of all non-faulty replicas in all subsequent views.

After every $N$ entries are added to the log ($N$ is a configurable constant), a replica $i$ broadcasts a $\tupstart\sync\com \textit{view-id}\com \textit{log-slot-num}\com \textit{drops}\tupend\sig{i}$ to all replicas, where \textit{log-slot-num} is latest log index that is a multiple of $N$, and \textit{drops} contains \textit{gap certificates} for all log slots that have been committed as \nop in the current view.
Once replica $i$ receives $2f$ \sync messages with the same \textit{log-slot-num} from different replicas, for each entry in any of the \textit{drops} that has a valid \textit{gap certificate}, it writes a \nop (possibly overwriting existing request) to the corresponding log position and saves the gap certificate.
It then updates its \textit{sync-point} to \textit{log-slot-num}.

%% file: proof.tex
\section{Correctness Proof}
\label{sec:proof}

This section contains complete safety and liveness proofs of the \sysname protocol.

\subsection{Safety}

The main safety property guaranteed by \sysname is linearizability.
In this safety proof, we assume the network primitive \net provides \textit{transferable authentication}, \textit{ordering}, and \textit{drop detection}, as specified in the main paper.

\begin{thm}[\sysname Safety]\label{thm:safety}
\sysname guarantees linearizability of client operations and returned results.
\end{thm}

Before proving \autoref{thm:safety}, we first define a few properties of \sysname replica logs and client \req{}s.

\begin{defn}
A \req is \textit{committed} in a log slot if it is executed by $2f+1$ distinct replicas in that slot with matching \textit{view-id} and \textit{log-hash}.
\end{defn}

\begin{defn}
A \req is \textit{successful} if the client receives $2f+1$ valid \reply{}s from different replicas with matching \textit{view-id}, \textit{log-slot-num}, \textit{log-hash}, and \textit{result}.
\end{defn}

It is easy to see that a successful \req implies that the \req is committed.

\begin{defn}
A log is \textit{stable} if it is a prefix of the log of every non-faulty replica in views higher than the current one.
\end{defn}

\begin{lem}[Log Stability]\label{lem:log-stability}
Every successful \req was appended onto a stable log at some non-faulty replica, and the resulting log is also stable.
\end{lem}

To prove \autoref{lem:log-stability}, we first prove the following set of lemmas.

\begin{lem}\label{lem:epoch-start-log-pos}
All non-faulty replicas that begin an epoch begin the epoch with the same log position.
\end{lem}
\begin{proof}
Prove by induction.
In the first epoch, all non-faulty replicas start with log position 0.
This proves the base case.
Now assume all non-faulty replicas start epoch $e$ with the same log position.
To enter the next epoch $e'$, a non-faulty replica needs to receive $2f+1$ \epochstart messages for $e'$ from distinct replicas with \textit{log-slot-num} matching its own.
Define these $2f+1$ \epochstart{}s as a epoch certificate.
By quorum intersection, no two non-faulty replicas can have epoch certificates with different \textit{log-slot-num}.
Therefore, all non-faulty replicas enter epoch $e'$ with the same log position.
This proves the inductive step.
\end{proof}

\begin{lem}\label{lem:no-op-req}
Within an epoch, if a \req is committed in some log slot $l$, then no replica can include a \nop with a valid gap certificate in slot $l$ in that epoch.
\end{lem}
\begin{proof}
We prove by contradiction.
Assume a replica inserts a \nop with a valid gap certificate in slot $l$.
The replica then has received $2f+1$ \gapcommit{}s with decision \textit{drop} for slot $l$, implying some replica has received $2f$ distinct \gapprepare{}s containing a \textit{drop} decision and a matching \gapdecision from the leader.
Since a non-faulty replica only sends unique \gapprepare for a log slot within a view, by quorum intersection, there cannot exist $2f$ distinct \gapprepare{}s containing a \textit{recv} decision at slot $l$.
Therefore, \textit{drop} is the only possible commit decision for the gap agreement protocol.
Moreover, a valid \gapprepare containing a \textit{drop} decision implies that $2f+1$ replicas have sent a \gapdropmsg.
Out of those $2f+1$ replicas, at least $f+1$ are non-faulty.
Since non-faulty replicas block until they receive \gapcommit and the only possible commit outcome is \textit{drop}, they will not execute \req.
By quorum intersection, no $2f+1$ replicas can execute \req, so \req cannot be committed.
This leads to a contradiction.
\end{proof}

\begin{lem}\label{lem:diff-req}
For any two non-faulty replicas in the same epoch, no slot in their logs contains different \req{}s.
\end{lem}
\begin{proof}
Within the epoch, non-faulty replicas insert \req{}s into their logs strictly in the order received from \net.
In the absence of \drop{}s, the \textit{ordering} property of \net ensures that all non-faulty replicas have identical sequence of \req{}s in their logs.
By \autoref{lem:epoch-start-log-pos}, for any epoch, all non-faulty replicas start the epoch with the same log position.
Consequently, for any two non-faulty replicas, no log slot within the epoch contains different \req{}s.
A non-faulty replica may also insert a \req $r$ into its log when handling a \drop.
To fill the gap caused by the \drop, the replica requires the transfer of $r$ with the corresponding ordering certificate \oc (through \queryreply or \gapcommit).
The \textit{transferable authentication} property of \net ensures that $r$ is identical to \req{}s delivered by other non-faulty replicas at the same position in the \net message sequence.
The case is therefore equivalent to the case in which a \req, not a \drop, is received by the replica.
A non-faulty replica may also insert a \nop into its log during the gap agreement protocol.
Assume the replica inserts \nop at the $l+i$th log slot where $l$ is the starting log position of the epoch.
The replica ignores the corresponding $i$th \req (by checking the sequence number) if it is later delivered by \net.
Consequently, if the replica receives the $i+1$th \req in the \net message sequence, it can only insert the \req at log position $l+i+1$. 
By the above argument, the replica's log from $l+i+1$ onward will not contain non-matching \req{}s from other non-faulty replicas.
By induction on $i$, no log slot within the epoch contains different \req{}s at any two non-faulty replicas.
\end{proof}

\begin{lem}\label{lem:commit-view-change}
Any \req or \nop that is committed at a log position in some view will be in the same log slot in all non-faulty replica's log in all subsequent views.
\end{lem}
\begin{proof}
To enter a new view $v'$ from the current view $v$, a non-faulty replica needs to receive a \viewstart message which contains $2f+1$ \viewchange messages from distinct replicas.
There are two cases.
Case 1: If a request $r$ is committed at log position $i$ in view $v$, by definition, $r$ is executed by $2f+1$ replicas at the same log position in $v$.
Therefore, at least $f+1$ non-faulty replicas have inserted $r$ into their log at position $i$.
By quorum intersection, at least one of the \viewchange messages contains $r$ in the log.
The log merging rule in the view change protocol ensures that the replica inserts $r$ into its log at the same log position in view $v'$.
And by \autoref{lem:no-op-req}, no \nop can be committed at the same log position, so the replica will not overwrite the slot with a \nop during log merging.
Case 2: If a \nop is committed at log position $i$ in view $v$, at least $2f+1$ replicas have sent \gapcommit with decision \drop for log slot $i$.
Therefore, at least $f+1$ non-faulty replicas have stored $2f$ \gapprepare and the matching \gapdecision with decision \drop.
By quorum intersection, at least one of the \viewchange messages contains the $2f$ \gapprepare and the \gapdecision.
The log merging procedure ensures that slot $i$ is filled with a \nop in view $v'$.
\end{proof}

We are now ready to prove the main log stability lemma.

\begin{proof}[Proof of Log Stability (\autoref{lem:log-stability})]
A successful \req implies that \req is committed at log slot $l$.
By definition of committed requests, at least $f+1$ non-faulty replicas have \textit{matching} logs up to $l$.
And since non-faulty replicas insert log entries strictly in log order (blocking before the next log entry is resolved), all log entries before $l$ are occupied.
For any log slot $l' \le l$, if a \req $r$ is stored in the matching logs, by \autoref{lem:diff-req}, no other \req can be inserted into the same slot at any non-faulty replica.
And by quorum intersection and our gap agreement protocol, there cannot exist a valid gap certificate for slot $l'$.
Therefore, only $r$ can be committed at log slot $l'$ in the view.
Otherwise, if a \nop is stored in the matching logs, there must exist a valid gap certificate for slot $l'$.
By definition, the \nop is committed at $l'$.
\autoref{lem:commit-view-change} then guarantees that log entry at $l'$ (either a \req or a \nop) in the matching log will be in the same log slot in all non-faulty replica's log in all subsequent views.
Consequently, the successful \req was appended onto a stable log at least $f+1$ non-faulty replica.
Since the successful \req is also committed at log slot $l$, by \autoref{lem:commit-view-change}, \req will be in slot $l$ in all non-faulty replica's log in all subsequent views.
The resulting log is thus also stable.
\end{proof}

With \autoref{lem:log-stability}, we are ready to prove our main safety property.

\begin{proof}[Proof of \sysname Safety (\autoref{thm:safety})]
First, observe that, by definition, a stable log only grows monotonically.
Combining this observation with \autoref{lem:log-stability}, from the client's perspective, the behavior of \sysname is indistinguishable from the behavior of a single, correct machine that processes \req{} sequentially.
This implies that any execution of \sysname is equivalent to some serial execution of \req{}s.
Moreover, clients retry sending an operation until a \req containing the operation is successful.
\sysname applies standard at-most-once techniques to avoid executing duplicated \req{}s.
Therefore, a \sysname execution is equivalent to some serial execution of unique client operations.

The above argument proves serializability.
When a client receives the necessary replies for a successful \req $r$, by \autoref{lem:log-stability}, $r$ must have already been added to a stable log.
For any successful \req $r'$ issued after this point in real-time, $r'$ can only be inserted after $r$ in the stable log.
Since non-faulty replicas execute \req{}s strictly in log order, the operations issued and results returned by \sysname are linearizable.
\end{proof}

\subsection{Liveness}

Due to the well-known FLP result, \sysname can not guarantee progress in a fully asynchronous network.
We therefore only prove liveness given some weak synchrony assumptions.

\begin{thm}[Liveness]\label{thm:liveness}
\req{}s sent by clients will eventually be successful if there is sufficient amount of time during which
\begin{itemize}
\setlength\itemsep{-3pt}
    \item the network the replicas communicate over is fair-lossy,
    \item there is some bound on the relative processing speeds of replicas,
    \item the $2f+1$ non-faulty replicas stay up,
    \item there is a non-faulty replica that stays up which no non-faulty replica suspects of having failed,
    \item there is a non-faulty \net sequencer stays up which no non-faulty replica suspects of having failed,
    \item all non-faulty replicas correctly suspect faulty nodes and \net sequencers,
    \item clients' \req{}s are eventually delivered by \net.
\end{itemize}
\end{thm}
\begin{proof}
Since there exist non-faulty replica and non-faulty \net sequencer that stay up which no non-faulty replica suspects of having failed, there is a finite number of view changes during the synchrony period.
Once the non-faulty replica that stays up has been elected as leader, and the non-faulty sequencer has been configured by the network, no view change with a higher view will start, as $2f+1$ non-faulty replicas will not send the corresponding \viewchange message.

Moreover, any view change that successfully starts will eventually finish.
A non-faulty replica that initiates a view change keeps re-broadcasting its \viewchange message until the new view starts, or until the view change is supplanted by one with a higher view.
Since non-faulty replicas correctly suspect faulty nodes and \net sequencers, eventually $2f+1$ non-faulty replicas will initiate the view change.
As all $2f+1$ non-faulty replicas stay up, the leader for the new view, if it is non-faulty, will eventually receive the necessary $2f+1$ \viewchange messages to start the view.
If the leader is faulty, non-faulty replicas will correctly suspect the fact, and start a higher view change which will supplant the current one.

Additionally, once a view starts, eventually all non-faulty replicas will adopt the new view and start processing \req{}s in the view, as long as the view is not supplanted by a even higher view.
If the leader is non-faulty, it will re-broadcast \viewstart messages until it receives acknowledgement from all replicas.
If the leader is faulty, non-faulty replicas will correctly suspect the fact, and start a higher view change which will supplant the current one.

The above arguments imply that eventually, there will be a view which stays active with a non-faulty leader and a non-faulty \net sequencer.
During that view, non-faulty replicas will eventually be able to resolve any \drop from \net:
The replica receiving a \drop will keep resending the \query message until receiving a \queryreply from the leader or enough \gapcommit{}s.
If the leader does not have the \req, it will continually broadcast \gapfindmsg to all replicas.
Since $2f+1$ non-faulty replicas stay up, eventually the leader will receive either one \gaprecvmsg or $2f+1$ \gapdropmsg{}s.
Once the non-faulty leader starts the binary Byzantine agreement protocol with the decision, by applying the same line of reasoning, eventually non-faulty replicas blocking on the \drop will receive the necessary \gapcommit{}s to resolve the \drop.

Therefore, the system will eventually reach a point where a view stays active with a non-faulty leader and a non-faulty \net sequencer, and non-faulty replicas only receive \req{}s from \net.
After that point, every \req delivered by \net will eventually be successful.
Because clients' \req{}s eventually will be delivered by \net and $2f+1$ non-faulty replicas stay up, clients will eventually receive the necessary \reply{}s for \req{}s they have sent.
\end{proof}